\documentclass[11pt]{article}
\usepackage{tabularx} 
\usepackage{graphicx}
\usepackage[graphicx]{realboxes}

\usepackage{adjustbox}
\usepackage{color,xcolor}
\usepackage{colortbl}
\usepackage{array}
\usepackage{bigstrut}

\usepackage{multirow}
\usepackage{mathtools}
\usepackage{morefloats}
\usepackage{amsfonts}
\usepackage{color}

\usepackage{latexsym}
\usepackage{amsmath,amssymb,amsthm}
\usepackage{dsfont}
\usepackage{extarrows}
\usepackage{hyperref}
\usepackage{lscape}
\usepackage{units}
\usepackage{amsmath}
\usepackage{natbib}

\numberwithin{equation}{section}

\newtheoremstyle{thm}
{9pt}
{9pt}
{\itshape}
{}
{\bfseries}
{.}
{ }
{}
\theoremstyle{thm}

\newtheorem{theorem}{Theorem}[section]

\newtheoremstyle{def}
{9pt}
{9pt}
{}
{}
{\bfseries}
{.}
{ }
{}
\theoremstyle{def}

\renewcommand{\footnoterule}{%
 & \kern -3.5pt
 & \hrule width \textwidth height 1pt
 & \kern 3.5pt
}

\makeatletter
\def\blfootnote{\xdef\@thefnmark{}\@footnotetext}
\makeatother

\begin{document}

\title{\bf On a new class of tests for the Pareto distribution using Fourier methods}

\author{L. Ndwandwe, J.S. Allison, M. Smuts, I.J.H. Visagie}

\date{\today}
\maketitle

\begin{abstract}
We propose new classes of tests for the Pareto type I distribution using the empirical characteristic function. These tests are $U$ and $V$ statistics based on a characterisation of the Pareto distribution involving the distribution of the sample minimum. In addition to deriving simple computational forms for the proposed test statistics, we prove consistency against a wide range of fixed alternatives. A Monte Carlo study is included in which the newly proposed tests are shown to produce high powers. These powers include results relating to fixed alternatives as well as local powers against mixture distributions. The use of the proposed tests is illustrated using an observed data set.
\end{abstract}


\section{Introduction}
\label{Sect1}

The Pareto distribution, nowadays commonly known as the Pareto type I distribution, was originally introduced by \cite{Par}. Since then several extensions of this distribution have been proposed. These extensions are achieved via the inclusion of a location, scale and inequality parameter, corresponding to the Pareto type II, III and IV distributions, respectively. Additionally, a so-called generalised Pareto distribution has been introduced. For an in depth discussion of the various types of Pareto distributions as well as the relationships between them, the interested reader is referred to \cite{Arn2015}.

The Pareto distribution is a popular model in engineering, economics, finance and actuarial science, especially where phenomena characterised by heavy tails are studied, see, e.g. \cite{Fis1961}, \cite{Ism2004}, \cite{nofal2017new}. Concrete examples of the use of the Pareto distribution include the modelling of failure times of mechanical components, see \cite{BSF2016}, as well as the modelling of excess of losses in insurance claims, see \cite{rytgaard1990estimation}. Due to its heavy tail, this distribution also plays a pivotal role in extreme value theory see, \cite{BeiGoeSegTeu}. Further examples of the use of the Pareto distribution can be found in \cite{Ami2007} and \cite{Sol2000}. A number of characterisations for the Pareto distribution have been developed in the literature, see, e.g. \cite{gupta1973characteristic} as well as \cite{Par}.

Due to the popularity of this distribution as well as its wide range of applications, goodness-of-fit tests have been developed in order to test the hypothesis that an observed dataset is compatible with the assumption of being realised from this distribution. For recent overview papers and discussions of some of these tests, see \cite{CDN2019} and \cite{ndwandwe2022testing} as well as the references therein. \cite{CDN2019} review tests for the generalised Pareto distribution as well as the Pareto types I and II, whereas \cite{ndwandwe2022testing} investigates several existing tests specifically for the Pareto type I distribution. Although tests exist for the Pareto type I distribution, they are few in number when compared to those for other distributions such as, for example, the normal or exponential distributions. In the remainder of this paper, we refer to the Pareto type I distribution simply as the Pareto distribution.

We propose new classes of goodness-of-fit tests for the Pareto distribution based on a characterisation involving the distribution of the sample minimum. In order to proceed, we introduce some notation. Let $X_1,\dots,X_n$ be independent and identically distributed (i.i.d.) realisations from a continuous random variable, $X$, with unknown distribution function $F$. Let $X_{1:n}<\dots<X_{n:n}$ denote the order statistics of $X_{1},\dots,X_{n}$. $X$ is said to follow the Pareto distribution with shape parameter $\beta$, denoted by $X \sim P(\beta)$, if it has distribution function
\begin{equation*}
    F(x) = 1-x^{-\beta}, \ \ \ x \ge 1,
\end{equation*}
for some $\beta>0$. The composite null hypothesis to be tested is
\begin{equation}
    H_{0}: X \sim P(\beta),\label{hyp}
\end{equation}
for some unspecified $\beta>0$. This hypothesis is tested against general alternatives. Throughout this paper, the value of $\beta$ is estimated by its method of moments estimator $\widehat{\beta}_n=\overline{X}_n/(\overline{X}_n-1)$, where $\overline{X}_n$ is the sample mean.

The remainder of this paper is structured as follows. In Section \ref{Sect2}, new classes of tests are proposed for the Pareto distribution, whereas Section \ref{Sect3} contains theoretical results pertaining to the asymptotic behaviour of the tests. A Monte Carlo study is presented in Section \ref{Sect4}, while Section \ref{Sect5} contains an example pertaining to observed data. The paper concludes in Section \ref{Sect6}.

\section{A new class of tests for the Pareto distribution}
\label{Sect2}

The proposed tests are based on a characterisation of the Pareto distribution via the distribution of the sample minimum. This characterisation is discussed in \cite{allison2021distribution} and is as follows:
\begin{theorem} \label{char}
Let $X,X_1,\dots,X_n$ be i.i.d. random variables from a continuous distribution with distribution function $F$. Let $m$ be an integer such that $2 \le m \le n$. $X^{1/m}$ and $\min(X_1,\dots,X_m)$ have the same distribution if, and only if,
\begin{equation*}
    F(x) = 1-x^{-\beta}, \ \ \ x \ge 1,
\end{equation*}
for some $\beta > 0$.
\end{theorem}
Since a random variable is characterised by its Fourier transform, we can base a test on the $V$ and $U$\textit{empirical characteristic functions} of $X^{1/m}$ and $\min(X_1, ..., X_m)$. Tests based on these quantities have been shown to not only possess desirable asymptotic properties, but also produce high powers in finite sample settings; the interested reader is referred to \cite{Mei2016}.

Let
\begin{equation*}
    \phi_m(t)=E\left[\textrm{e}^{itX^{1/m}}\right]
\end{equation*}
and
\begin{equation*}
    \xi_m(t)=E\left[\textrm{e}^{it\min(X_{1},\dots,X_{m})}\right]
\end{equation*}
be the characteristic functions of $X^{1/m}$ and $\min(X_1,\dots,X_m)$, respectively. Denote the empirical versions of $\phi_m$ and $\xi_m$ by
\begin{equation*}
    \phi_{n,m}(t) = \frac{1}{n} \sum_{j=1}^n \textrm{e}^{itX_{j:n}^{1/m}}
\end{equation*}
and
\begin{equation*}
\xi_{n,m}(t)  = \frac{1}{n^m} \sum_{k_1=1}^n \cdots \sum_{k_m=1}^n \textrm{e}^{it\textrm{min}(X_{k_1},\dots,X_{k_m})}.
\end{equation*}
Theorem \ref{char} implies that, for all $t\in \mathbb{R}$, $\phi_m(t)=\xi_m(t)$ if, and only if, $X \sim P(\beta)$ for some $\beta>0$.
We propose a class of tests for the hypothesis in (\ref{hyp}) based on a weighted $L2$ distance between $\phi_{n,m}$ and $\xi_{n,m}$:
\begin{equation*}
    S_{n,m,a} = n \int_{-\infty}^{\infty} |\phi_{n,m}(t)-\xi_{n,m}(t)|^2 w_a(t)\textrm{d}t,
\end{equation*}
where $w_a(t)$ is an appropriate weight function depending on a user defined parameter $a$. This weight function is included in order to ensure the existence of the integral. We choose $w_a$ such that $\int_{-\infty}^{\infty} w_a(t) dt < \infty$. Popular choices of $w_a$ include the Laplace and Gaussian kernels.

Note that
\begin{equation}
    S_{n,m,a} 
=\frac{1}{n^{2m}}\sum_{k_1=1}^n \cdots \sum_{k_{2m}=1}^n h(X_{k_1},\dots,X_{k_{2m}};a), \label{Vstatformula}
\end{equation}
where
\begin{eqnarray*}
h(X_{k_1},\dots,X_{k_{2m}};a)&=& \hspace{-1mm} \int_{-\infty}^{\infty}\left[\cos \left(t(X_{k_1}^{1/m}-X_{k_2}^{1/m})\right) -2 \cos \left(t (X_{k_1}^{1/m}-\min(X_{k_1},\dots,X_{k_m}))\right)\right.\\
&+&\left. \cos \left(t(\min(X_{k_1},\dots,X_{k_m})-\min(X_{k_{m+1}},\dots,X_{k_{2m}})) \right) \right]w_a(t)\textrm{d}t.
\end{eqnarray*}
Therefore, $S_{n,m,a}$ is a $V$ statistic of order $2m$ with kernel $h$. The form of $S_{n,m,a}$ specified in (\ref{Vstatformula}) is computationally expensive (e.g., if $m=4$, then computing $S_{n,m,a}$ requires the evaluation of an eight fold summation). However, after some combinatorics and algebraic manipulation, $\xi_{n,m}(t)$ can be expressed as a single sum in terms of the order statistics:
\begin{equation*}
    \xi_{n,m}(t) = \sum_{j=1}^{n} v_{j,m}e^{itX_{j:n}},
\end{equation*}
with
\begin{equation*}
    v_{j,m}:=\frac{1}{n^m}\left[ (n-j+1)^m- (n-j)^m\right].
\end{equation*}

When using a Laplace kernel as weight function, $w_a(t)=e^{-a|t|}$, we denote the resulting test statistic by $S^{(1)}_{n,m,a}$;
\begin{align*}
S_{n,m,a}^{(1)} & =\frac{1}{n}\sum_{j=1}^{n}\sum_{k=1}^{n}\left\{\frac{2a}{a^{2}+\left(X_{j:n}^{1/m}-X_{k:n}^{1/m}\right)^{2}}-nv_{j,m}\frac{4a}{a^{2}+\left(X_{j:n}-X_{k:n}^{1/m}\right)^{2}}\right.\\
&\left.+ n^2v_{j,m}v_{k,m}\frac{2a}{a^{2}+\left(X_{j:n}-X_{k:n}\right)^{2}}\right\}.
\end{align*}
Upon setting the weight function equal to a Gaussian kernel, $\widetilde{w}_a(t)=e^{-at^2}$, we obtain $S^{(2)}_{n,m,a}$;
\begin{align*}
S_{n,m,a}^{(2)} & =\frac{1}{n}\sqrt{\frac{\pi}{a}}\sum_{j=1}^{n}\sum_{k=1}^{n}\left\{\exp\left[\frac{-\left(X_{j:n}^{1/m}-X_{k:n}^{1/m}\right)^{2}}{4a}\right]-2nv_{j,m}\exp\left[\frac{-\left(X_{j:n}-X_{k:n}^{1/m}\right)^{2}}{4a}\right]\right.\\
&\left.+ n^2v_{j,m}v_{k,m}\exp\left[\frac{-\left(X_{j:n}-X_{k:n}\right)^{2}}{4a}\right]\right\}.
\end{align*}

Above we consider $S_{n,m,a}$, based on $V$ statistics. We now turn our attention to the situation where the empirical characteristic functions are estimated using $U$ statistics. Denote the difference between the $U$ \textit{empirical characteristic functions} of $X^{1/m}$ and $\textrm{min}({X_1,...,X_m})$ by
\begin{equation*}
    D_{n,m}(t)=\phi_{n,m}(t) -\psi_{n,m}(t),
\end{equation*}
where
\begin{equation*}
    \psi_{n,m}(t) = \binom{n}{m}^{-1} \sum_{1 \le k_1<\dots<k_m \le n} \textrm{e}^{it\textrm{min}(X_{k_1},\dots,X_{k_m})}
\end{equation*}
is the empirical characteristic function of the $\binom{n}{m}$ random variables $\textrm{min}(X_{k_1},\dots,X_{k_m})$, $1 \le k_1<\dots<k_m \le n$.

After some algebra it follows that $\psi_{n,m}(t)$ can be expressed as a single summation;
\begin{equation*}
    \psi_{n,m}(t) = \binom{n}{m}^{-1} \ \ \sum_{j=1}^{n-m+1} u_{j,m}\textrm{e}^{itX_{j:n}},
\end{equation*}
where 
\begin{equation*}
    u_{j,m}= \binom{n-j}{m-1}.
\end{equation*}
From Theorem \ref{char} it follows that, if $X_1,\dots,X_n$ is a random sample from the Pareto distribution, then the difference between $\phi_{n,m}(t)$ and $\psi_{n,m}(t)$ should be close to zero. We thus suggest the test statistic
\begin{equation}
    T_{n,m,a} = n \int_{-\infty}^{\infty} |\phi_{n,m}(t)-\psi_{n,m}(t)|^2 w_a(t)\textrm{d}t. \label{five}
\end{equation}

After some algebra, we obtain the following easily calculable expression for the test statistic based on the choices $w_a(t)=e^{-a|t|}$ and $\widetilde{w}_a(t)=e^{-at^2}$, respectively:
\begin{align*}
T_{n,m,a}^{(1)} & =\frac{1}{n}\sum_{j=1}^{n}\sum_{k=1}^{n}\frac{2a}{a^{2}+\left(X_{(j)}^{1/m}-X_{(k)}^{1/m}\right)^{2}}-\sum_{j=1}^{n-m+1}\sum_{k=1}^{n}\binom{n}{m}^{-1}u_{j,m}\frac{4a}{a^{2}+\left(X_{(j)}-X_{(k)}^{1/m}\right)^{2}}\\
 & +n\sum_{j=1}^{n-m+1}\sum_{k=1}^{n-m+1}\binom{n}{m}^{-2}u_{j,m}u_{k,m}\frac{2a}{a^{2}+\left(X_{(j)}-X_{(k)}\right)^{2}},
\end{align*}
and
\begin{align*}
T_{n,m,a}^{(2)} & =\sqrt{\frac{\pi}{a}}\left\{ \frac{1}{n}\sum_{j=1}^{n}\sum_{k=1}^{n}\exp\left[\frac{-\left(X_{(j)}^{1/m}-X_{(k)}^{1/m}\right)^{2}}{4a}\right]\right.\\
 & \left.-2\sum_{j=1}^{n-m+1}\sum_{k=1}^{n}\binom{n}{m}^{-1}u_{j,m}\exp\left[\frac{-\left(X_{(j)}-X_{(k)}^{1/m}\right)^{2}}{4a}\right]\right.\\
 & \left.+n\sum_{j=1}^{n-m+1}\sum_{k=1}^{n-m+1}\binom{n}{m}^{-2}u_{j,m}u_{k,m}\exp\left[\frac{-\left(X_{(j)}-X_{(k)}\right)^{2}}{4a}\right]\right\}. 
\end{align*}


\section{Consistency of the tests}
\label{Sect3}

In this section we only present the results pertaining to $T_{n,m,a}$; the derivations relating to $S_{n,m,a}$ follows from analogous arguments and are therefore omitted for the sake of brevity. Before proceeding to prove the consistency of $T_{n,m,a}$, some comments about the asymptotic null distribution of the test statistic are in order.

$T_{n,m,a}$ is formulated as a weighted $L2$-type statistic involving empirical characteristic functions. The asymptotic null distribution of these classes of statistics are studied in, amongst others, \cite{feuerverger1977empirical}, \cite{baringhaus1988consistent}, \cite{klar2005tests} as well as \cite{baringhaus2017limit}. The asymptotic null distribution of $T_{n,m,a}$ will typically correspond to that of $\int_{-\infty}^{\infty} |V(t)|^2w_a(t)dt =:T_{m,a}$, where $V(\cdot)$ is a Gaussian process with zero-mean. $T_{m,a}$ has the same distribution as $\sum_{j=1}^{\infty}\lambda_j\chi_j^2$, where $\chi_j^2$ are i.i.d random variables following a chi-squared distribution with one degree of freedom. However, the covariance matrix of $V(\cdot)$ as well as the eigenvalues $\lambda_j$ depend on the unknown underlying distribution $F$, usually in a complicated way. We therefore make use of a parametric bootstrap procedure in order to estimate the critical values of these tests (see Section \ref{simsetting}).

The following theorem is concerned with the asymptotic behaviour of $T_{n,m,a}$ under fixed alternative distributions. 
\begin{theorem}
Let $X_1,...,X_n$ be independent copies of a continuous random variable $X$ with finite mean, then
\begin{equation} \label{bc}
    \frac{T_{n,m,a}}{n}\overset{p}{\longrightarrow}\Delta_{m,w} := \int_{-\infty}^{\infty} |\phi_m(t)-\psi_m(t)|^2 w_a(t)\textrm{d}t \geq 0,
\end{equation}
as $n \rightarrow \infty$, with $\Delta_{m,w} =0$ if, and only if, $X \sim P(\beta)$.
\end{theorem}

\begin{proof}
Recall from (\ref{five}) that
\begin{equation*}
  \frac{T_{n,m,a}}{n}=\int_{-\infty}^{\infty}D_{n,m}^2(t)w_a(t)\textrm{d}t, 
\end{equation*}
where $D_{n,m}^2(t)=|\phi_{n,m}(t)-\psi_{n,m}(t)|^2$. We have, by the law of large numbers, that $\phi_{n,m}(t)\overset{p}{\longrightarrow}\phi_m(t)$ and that $\psi_{n,m}(t)\overset{p}{\longrightarrow}\psi_m(t)$. By the continuous mapping theorem it follows that $D_{n,m}^2(t)\overset{p}{\longrightarrow}D_m^2(t):=|\phi_{m}(t)-\psi_{m}(t)|^2$. Now, $D_{n,m}^2(t)\leq4$, hence an application of Lebesgue's theorem of dominated convergence yields \eqref{bc}. In view of the characterisation given in Theorem \ref{char}, it follows that $\Delta_{m,w}$ is zero if, and only if $X\sim P(\beta)$
\end{proof}

\section{Monte Carlo study}
\label{Sect4}

In this section Monte Carlo simulations are used to compare the finite sample performance of the newly proposed tests to the following five existing goodness-of-fit tests for the Pareto distribution:
\begin{itemize}
\item The traditional Kolmogorov-Smirnov ($KS_{n}$), Cram\'{e}r-von Mises ($CM_n$) and Anderson-Darling $(AD_n)$ tests.
\item A test, proposed by \cite{zhang2002powerful}, based on the likelihood ratio.
The test statistic is given by
\[
ZA_{n} =\int_{-\infty}^{\infty}G_n^{2}(t)\left[F_{n}(t)\left(1-F_{n}(t)\right)\right]^{-1}dF_{n}(t),
\]
where 
\[
G_n^{2}(t)=2n\left\{ F_{n}(t)\log\left(\frac{F_{n}(t)}{F(t,\hat{\beta}_{n})}\right)+\left[1-F_{n}(t)\right]\log\left(\frac{1-F_{n}(t)}{1-F(t,\hat{\beta}_{n})}\right)\right\} 
\]
is the likelihood ratio statistic and $F_n(t) = \frac{1}{n} \sum_{j=1}^n \textrm{I}(X_j \le t)$.
The computational form of the test statistic is
\begin{align*}
ZA_{n} & =-\sum_{j=1}^{n}\left[\frac{\log\left\{ 1-X_{j:n}^{-\hat{\beta}_{n}}\right\} }{n-j+\frac{1}{2}}+\frac{\log\left\{ X_{j:n}^{-\hat{\beta}_{n}}\right\} }{j-\frac{1}{2}}\right].
\end{align*}

\item A test based on Mellin transform proposed by \cite{meintanis2009unified}. The test statistic is given by 
\begin{eqnarray*}
    G_{n,w}&=&\frac{1}{n}\left[(\widehat{\beta}_n+1)^2\sum_{j,k=1}^{n}I^{(0)}_w({X_jX_k})+\sum_{j,k=1}^{n}I^{(2)}_w({X_jX_k})+2(\widehat{\beta}_n+1)\sum_{j,k=1}^{n}I^{(1)}_w({X_jX_k})\right]\\
    &&+\widehat{\beta}_n\left[n\widehat{\beta}_nI^{(0)}_w(1)-2(\widehat{\beta}_n+1)\sum_{j=1}^{n}I^{(0)}_w({X_j})-2\sum_{j=1}^{n}I^{(1)}_w({X_j})\right],
   \end{eqnarray*}
   where 
\begin{equation*}
    I^{(m)}_w(t)=\int_{0}^{\infty}(t-1)^m\frac{1}{x^t}w(t)dt, \quad m=0,1,2.
\end{equation*}
Choosing $w(x)=e^{-ax}$, one has
\begin{equation*}
    I^{(0)}_a(x)=(a+\log x)^{-1},
\end{equation*}
\begin{equation*}
    I^{(1)}_a(x)=\frac{1-a-\log x}{(a+\log x)^2},
\end{equation*}
and
\begin{equation*}
    I^{(2)}_a(x)=\frac{2-2a+a^2+2(a-1)\log x +\log^2 x}{(a+\log x)^3}.
\end{equation*}
The value of the tuning parameter $a$ is set to $1$ in order to obtain the numerical results presented.
\end{itemize}

\subsection{Simulation setting} \label{simsetting}

Power (and size) estimates are calculated at a significance level of 5\% for sample sizes $n=20$ and $n=30$ using 50 000 independent Monte Carlo replications. Since the null distributions of the test statistics depend on the value of the unknown shape parameter $\beta$, we use a parametric bootstrap procedure to calculate numerical critical values. For computational efficiency, we employ the warp-speed bootstrap methodology proposed by \cite{giacomini2013warp}. This methodology is outlined in the following algorithm:
\begin{enumerate}
\item Draw a sample of size $n$, say $X_{1},\dots,X_{n}$ from an alternative distribution and estimate the parameter
$\beta$ by $\widehat{\beta}_n=\overline{X}_n/(\overline{X}_n-1)$.
\item Calculate the value of the test statistic say $S=S_{n}(X_{1},\dots,X_{n})$.
\item Generate a bootstrap sample $X_{1}^{*},\dots,X_{n}^{*}$
by independently sampling from a $P(\widehat{\beta}_{n})$ distribution.
Calculate the value of the test statistic using the bootstrap sample,
$S^{*}=S_{n}(X_{1}^{*},\dots,X_{n}^{*})$.
\item Repeat steps 1--3, MC times to obtain $S_{1},\dots,S_{MC}$ and $S_{1}^{*},\dots,S_{MC}^{*}$,
where $S_{j}$ and $S_{j}^{*}$ denote the values of the test statistic
for the $j^{th}$ sample generated in Steps 2 and 3, respectively.
\item The power estimate is given by $\frac{1}{MC}\sum_{j=1}^{MC}\textrm{I}\left(S_{j}>S_{\left\lfloor MC(1-\alpha)\right\rfloor:MC}^{*}\right)$
for $j=1,\dots,MC$, where $S_{j:MC}^{*}$ denotes the $j^{th}$ order statistic of $S_{1}^{*},\dots,S_{MC}^{*}$, where $\left\lfloor \cdot\right\rfloor $
denotes the floor function and $\textrm{I}(\cdot)$ denotes the indicator function.
\end{enumerate}

The simulation study presented considers two sets of power results. The first is concerned with powers against the fixed alternative distributions specified in Table \ref{tab1}. The resulting empirical powers for sample sizes of $n=20$ and $n=30$ can be found in Tables \ref{tab2} and \ref{tab3}, respectively. Second, we consider some local power estimates where we simulate data from two families of mixture distributions. In the first of the mixture distributions used, we simulate from a $LN(1)$ with probability $p$, and from a Pareto distribution (with the same mean as the $LN(1)$) with probability $1-p$; the empirical powers obtained are reported in Table \ref{tab4}. The second family of mixture distributions is obtained upon replacing the $LN(1)$ distribution by the exponential distribution with mean $0.5$; the calculated powers can be found in Table \ref{tab5}. The results shown in Tables \ref{tab4} and \ref{tab5} include two powers for each listed distribution; the first is associated with a sample size of $20$, while the second is the estimated power based on a sample of size $30$.

\begin{table}
\begin{center}
\small
\begin{tabular}{lll}
\hline
.Alternative & Density function & Notation \\
\hline
Gamma & $\displaystyle\frac{1}{\Gamma(\theta)}(x-1)^{\theta-1}e^{-(x-1)}$ & $\Gamma(\theta)$\\
Weibull & $\displaystyle\theta (x-1)^{\theta-1} \exp\left\{- (x-1)^\theta\right\}$& W($\theta$)\\
%
Lognormal & $\displaystyle\exp\left( -\frac{1}{2}\left(\log(x-1)/\theta\right)^2 \right) \left/ \left\{\theta (x-1) \sqrt{2\pi} \right\} \right.$& LN($\theta$)\\
Linear failure rate & $\displaystyle(1+\theta (x-1))\exp(-(x-1) - \theta (x-1)^2 / 2)$& LF($\theta$) \\

%
Beta-exponential & $\displaystyle\theta e^{-x}(1-e^{-x})^{\theta-1}$  &BEX($\theta$)\\
Dhillon & $\frac{\theta + 1}{x+1}\exp\left\{-(\log(x+1))^{\theta+1}\right\}(\log(x+1))^\theta$ & $DH(\theta)$\\
Log-normal & $\exp\left\{-\frac{1}{2}(\log(x-1)/\theta)^2\right\}/\left\{\theta (x-1) \sqrt{2\pi}\right\}$ & $LN(\theta)$ \\
Half-normal & $\frac{\sqrt{2}}{\theta\sqrt{\pi}}\exp\left(\frac{-(x-1)^2}{2\theta^2}\right)$ & $HN(\theta)$ \\
\hline
\end{tabular}
\end{center}
\caption{Alternative distributions used}
\label{tab1}
\end{table}

The reported empirical powers of the new tests were obtained by setting $m=3$ and $a=2$ in all instances. Several other values for these parameters were considered when performing the Monte Carlo study; however, the specified choices generally resulted in high powers. For the sake of brevity, we omit the results pertaining to other parameter configurations and only display those associated with $m=3$ and $a=2$. All calculations were performed in R; see \cite{rteam}.

\subsection{Simulation results}

The power estimates in Tables \ref{tab2} to \ref{tab5} are the percentages (rounded to the nearest integer) of the number of samples resulting in a rejection of the null hypothesis. For ease of comparison the highest two powers in each row (including ties) are printed in bold.

\begin{table}[!htbp] \centering 
  \caption{Empirical powers against fixed alternatives for $n=20$} 
  \label{} 
  \resizebox{\columnwidth}{!}{
\begin{tabular}{@{\extracolsep{13pt}} lccccccccccc} 
\\[-1.8ex]\hline 
\hline \\[-1.8ex] 
Dist & $KS_n$ & $CV_n$ & $AD_n$ & $ZA_n$ & $G_{n,2}$ & $S^{(1)}_{n,3,2}$ & $S^{(2)}_{n,3,2}$ &  $T^{(1)}_{n,3,2}$  & $T^{(2)}_{n,3,2}$\\
\hline \\[-1.8ex] 
$P(2)$  &    \textbf{5} & \textbf{5} & \textbf{5} & \textbf{5} & \textbf{5} & \textbf{5} & \textbf{5} & \textbf{5} & \textbf{5} \\
$P(5)$  &    \textbf{5} & \textbf{5} & \textbf{5} & \textbf{5} & \textbf{5} & 4     & 4     & 4     & 4 \\
$P(10)$  &    \textbf{5} & \textbf{5} & \textbf{5} & \textbf{5} & \textbf{5} & 4     & 4     & 4     & 4 \\
$\Gamma{(0.5)}$  &  15    & 15    & \textbf{37} & \textbf{38} & 9     & 4     & 3     & 11    & 10 \\
$\Gamma{(0.8)}$  &    15    & 17    & 15    & 11    & 16    & \textbf{21} & \textbf{21} & 16    & 15 \\
$\Gamma{(1)}$  &    38    & 45    & 43    & 34    & 47    & \textbf{50} & \textbf{51} & 43    & 42 \\
$\Gamma{(1.2)}$  &    65    & 74    & 72    & 64    & 76    & \textbf{77} & \textbf{78} & 71    & 71 \\
$W(0.5)$     &    16    & 15    & \textbf{43} & \textbf{45} & 13    & 19    & 16    & 30    & 26 \\
$W(0.8)$     &   15    & 16    & 16    & 10    & 14    & \textbf{20} & \textbf{20} & 16    & 14 \\
$W(1.2)$     &    65    & 75    & 73    & 65    & 78    & \textbf{79} & \textbf{79} & 73    & 72 \\
$W(1.5)$     &    91    & 96    & 96    & 93    & \textbf{97} & \textbf{97} & \textbf{97} & 95    & 95 \\
$LN(1)$    &    73    & 82    & 83    & \textbf{90} & \textbf{86} & 84    & \textbf{86} & 80    & 81 \\
$LN(1.2)$    &    42    & 51    & 52    & \textbf{59} & \textbf{55} & 51    & 54    & 45    & 46 \\
$LN(1.5)$    &    17    & 21    & \textbf{22} & \textbf{22} & 20    & 19    & 21    & 17    & 17 \\
$LN(2.5)$    &    11    & 9     & 25    & 22    & 4     & 31    & 35    & \textbf{43} & \textbf{47} \\
$LFR(0.2)$    &    45    & 53    & 50    & 40    & 55    & \textbf{59} & \textbf{59} & 50    & 50 \\
$LFR(0.5)$    &    51    & 59    & 56    & 46    & 62    & \textbf{65} & \textbf{65} & 57    & 56 \\
$LFR(0.8)$    &    54    & 62    & 59    & 50    & 66    & \textbf{68} & \textbf{68} & 60    & 60 \\
$LFR(1)$    &    57    & 66    & 62    & 53    & 69    & \textbf{71} & \textbf{71} & 63    & 63 \\
$BE(0.5)$    &    12    & 11    & \textbf{32} & \textbf{34} & 6     & 5     & 3     & 10    & 7 \\
$BE(0.8)$    &    17    & 20    & 18    & 13    & 19    & \textbf{24} & \textbf{24} & 19    & 18 \\
$BE(1)$    &    39    & 45    & 43    & 34    & 47    & \textbf{51} & \textbf{51} & 43    & 42 \\
$BE(1.5)$    &    84    & 91    & 91    & 86    & \textbf{93} & 92    & \textbf{93} & 89    & 89 \\
$D(0.2)$    &    26    & 31    & \textbf{32} & 29    & 30    & 30    & \textbf{32} & 26    & 26 \\
$D(0.4)$    &    50    & 59    & 59    & 57    & \textbf{60} & 58    & \textbf{61} & 52    & 53 \\
$D(0.6)$    &    73    & 82    & 82    & 80    & \textbf{84} & 82    & \textbf{84} & 77    & 79 \\
$D(0.8)$    &    89    & 94    & 94    & 93    & \textbf{95} & 94    & \textbf{95} & 92    & 93 \\
$HN(0.8)$    &    59    & 68    & 65    & 56    & 71    & \textbf{72} & \textbf{72} & 65    & 64 \\
$HN(1)$    &    66    & 75    & 73    & 62    & 78    & \textbf{79} & \textbf{79} & 72    & 72 \\
\hline \\[-1.8ex] 
\end{tabular}
\label{tab2}
}
\end{table} 

\begin{table}[!htbp] \centering 
  \caption{Empirical powers against fixed alternatives for $n=30$} 
  \label{} 
  \resizebox{\columnwidth}{!}{
\begin{tabular}{@{\extracolsep{13pt}} lccccccccccc} 
\\[-1.8ex]\hline 
\hline \\[-1.8ex] 
Dist & $KS_n$ & $CV_n$ & $AD_n$ & $ZA_n$ & $G_{n,2}$ &  $S^{(1)}_{n,3,2}$  & $S^{(2)}_{n,3,2}$  & $T^{(1)}_{n,3,2}$ &  $T^{(2)}_{n,3,2}$\\
\hline \\[-1.8ex] 
$P(2)$     &    \textbf{5} & \textbf{5} & \textbf{5} & \textbf{5} & \textbf{5} & \textbf{5} & \textbf{5} & \textbf{5} & \textbf{5} \\
$P(5)$     &    \textbf{5} & \textbf{5} & \textbf{5} & \textbf{5} & \textbf{5} & 4     & 4     & 4     & 4 \\
$P(10)$     &    \textbf{5} & \textbf{5} & \textbf{5} & \textbf{5} & \textbf{5} & 4     & 4     & 4     & 4 \\
$\Gamma{(0.5)}$     &    21    & 21    & \textbf{49} & \textbf{52} & 11    & 10    & 7     & 21    & 17 \\
$\Gamma{(0.8)}$     &    19    & 22    & 20    & 14    & 21    & \textbf{26} & \textbf{26} & 20    & 19 \\
$\Gamma{(1)}$     &    52    & 60    & 58    & 45    & 61    & \textbf{64} & \textbf{65} & 57    & 57 \\
$\Gamma{(1.2)}$     &    82    & 90    & 89    & 81    & \textbf{91} & 90    & \textbf{92} & 87    & 87 \\
$W(0.5)$     &    23    & 21    & \textbf{54} & \textbf{60} & 16    & 35    & 28    & 48    & 41 \\
$W(0.8)$     &    19    & 21    & 21    & 12    & 18    & \textbf{24} & \textbf{24} & 20    & 18 \\
$W(1.2)$     &    83    & 91    & 90    & 84    & \textbf{93} & 92    & \textbf{93} & 89    & 89 \\
$W(1.5)$     &    98    & \textbf{100} & \textbf{100} & 99    & \textbf{100} & \textbf{100} & \textbf{100} & \textbf{100} & \textbf{100} \\
$LN(1)$    &    88    & 94    & 95    & \textbf{98} & \textbf{96} & 95    & \textbf{96} & 93    & 94 \\
$LN(1.2)$    &    55    & 66    & 68    & \textbf{79} & \textbf{71} & 65    & 69    & 60    & 62 \\
$LN(1.5)$    &    22    & 26    & \textbf{29} & \textbf{30} & 26    & 22    & 25    & 19    & 20 \\
$LN(2.5)$    &    13    & 10    & 31    & 31    & 4     & 47    & 51    & \textbf{58} & \textbf{61} \\
$LFR(0.2)$    &    62    & 71    & 69    & 57    & 73    & \textbf{74} & \textbf{75} & 68    & 68 \\
$LFR(0.5)$    &    69    & 78    & 76    & 65    & 80    & \textbf{81} & \textbf{81} & 75    & 75 \\
$LFR(0.8)$    &    73    & 81    & 80    & 69    & \textbf{84} & \textbf{85} & \textbf{84} & 79    & 79 \\
$LFR(1)$    &    74    & 83    & 81    & 72    & 85    & \textbf{86} & \textbf{86} & 81    & 80 \\
$BE(0.5)$    &    15    & 15    & \textbf{41} & \textbf{46} & 7     & 9     & 5     & 17    & 12 \\
$BE(0.8)$    &    24    & 26    & 25    & 16    & 25    & \textbf{31} & \textbf{31} & 25    & 24 \\
$BE(1)$    &    53    & 62    & 60    & 46    & 63    & \textbf{65} & \textbf{66} & 58    & 58 \\
$BE(1.5)$    &    95    & 98    & 98    & 97    & \textbf{99} & \textbf{99} & \textbf{99} & 98    & 98 \\
$D(0.2)$    &    34    & \textbf{41} & \textbf{41} & 39    & 40    & 37    & 39    & 33    & 33 \\
$D(0.4)$    &    67    & 75    & \textbf{76} & 73    & \textbf{77} & 73    & \textbf{76} & 68    & 70 \\
$D(0.6)$    &    89    & 94    & 94    & 93    & \textbf{95} & 93    & \textbf{95} & 91    & 93 \\
$D(0.8)$    &    97    & \textbf{99} & \textbf{99} & \textbf{99} & \textbf{99} & \textbf{99} & \textbf{99} & 98    & \textbf{99} \\
$HN(0.8)$    &    78    & 86    & 84    & 75    & \textbf{88} & \textbf{88} & \textbf{88} & 83    & 83 \\
$HN(1)$    &    85    & 91    & 90    & 82    & \textbf{92} & \textbf{92} & \textbf{92} & 89    & 88 \\
\end{tabular}
\label{tab3}
}
\end{table} 

\begin{table}[!htbp] \centering 
  \caption{Empirical powers against lognormal mixtures} 
  \label{} 
\begin{tabular}{@{\extracolsep{1pt}} cccccccccccccc} 
\hline 
$p$ & $KS_n$ & $CV_n$ & $AD_n$ & $ZA_n$ & $G_{n,2}$ & $S^{(1)}_{n,3,2}$  & $S^{(2)}_{n,3,2}$ & $T^{(1)}_{n,3,2}$ & $T^{(2)}_{n,3,2}$\\
\hline

\multirow{2}{12pt}{$0.0$} &\textbf{5} & \textbf{5} & \textbf{5} & \textbf{5} & \textbf{6} & \textbf{5} & \textbf{5} & \textbf{5} & \textbf{5} \\
 &\textbf{5} & \textbf{5} & \textbf{5} & \textbf{5} & \textbf{5} & \textbf{5} & \textbf{5} & \textbf{5} & \textbf{5}  \\

 \multirow{2}{12pt}{$0.1$} &   \textbf{5} & \textbf{5} & \textbf{5} & \textbf{5} & \textbf{5} & \textbf{5} & \textbf{5} & \textbf{5} & 4\rule{0pt}{20pt}\\
 &    \textbf{5} & \textbf{5} & \textbf{5} & \textbf{5} & \textbf{5} & \textbf{5} & \textbf{5} & \textbf{5} & \textbf{5}  \\

\multirow{2}{12pt}{$0.2$} &    6     & 6     & 5     & 6     & 6  & \textbf{7} & \textbf{7}   & 6     & 5\rule{0pt}{20pt} \\
 &    6     & 6     & 5     & 6     & 6     & \textbf{7} & \textbf{7} & 6     & 5 \\

\multirow{2}{12pt}{$0.3$} &    7     & 8     & 7     & 7     & 8     & \textbf{10} & \textbf{9}   & 7     & 7\rule{0pt}{20pt} \\
 &    8     & 8     & 8     & 8     & 9    &  \textbf{11} & \textbf{11} & 9     & 8 \\

\multirow{2}{12pt}{$0.4$} &    10    & 11    & 10    & 10    & 12    & \textbf{14} & \textbf{14} & 11   & 10\rule{0pt}{20pt} \\
 &    12    & 13    & 12    & 11    & 14    & \textbf{16}    & \textbf{17} & 13    & 13    \\

 \multirow{2}{12pt}{$0.5$} &    15    & 17    & 15    & 15    & 18  & \textbf{20} & \textbf{20} & 16    & 15\rule{0pt}{20pt} \\
 &    18    & 22    & 21    & 18    & 23    & \textbf{26} & \textbf{26} & 21    & 20 \\

\multirow{2}{12pt}{$0.6$} &   22    & 26    & 24    & 23    & 28    & \textbf{29} & \textbf{30} & 24    & 23\rule{0pt}{20pt} \\
 &    29    & 34    & 32    & 28    & 35    & \textbf{38} & \textbf{39} & 32    & 32 \\

\multirow{2}{12pt}{$0.7$} &    31    & 37    & 35    & 33    & 39   & \textbf{41}    & \textbf{42}     & 35    & 34\rule{0pt}{20pt} \\
 &    43    & 50    & 48    & 44    & 52     & \textbf{53}   & \textbf{55}     & 47    & 48 \\

\multirow{2}{12pt}{$0.8$} &    43    & 51    & 50    & 49    & 54    & \textbf{55} & \textbf{56} & 48    & 48\rule{0pt}{20pt} \\
 &    59    & 67    & 67    & 63    & \textbf{70}    & \textbf{70}    & \textbf{73}     & 65    & 66 \\

\multirow{2}{12pt}{$0.9$} &    57    & 66    & 66    & 68    & \textbf{71} & 70    & \textbf{71}& 64    & 64\rule{0pt}{20pt} \\
 &    75    & 83    & 83    & 83    & \textbf{86} & 85    & \textbf{87} & 81    & 82 \\

\hline \hline
\label{tab4}
\end{tabular} 
\end{table}


\begin{table}[!htbp] \centering 
  \caption{Empirical powers against exponential mixtures}
  \label{} 
\begin{tabular}{@{\extracolsep{1pt}} cccccccccccccc} 
\hline 
$p$ & $KS_n$ & $CV_n$ & $AD_n$ & $ZA_n$ & $G_{n,2}$ & $S^{(1)}_{n,3,2}$  & $S^{(2)}_{n,3,2}$ & $T^{(1)}_{n,3,2}$ & $T^{(2)}_{n,3,2}$\\
\hline
\multirow{2}{12pt}{$0.0$} & \textbf{5} & \textbf{5} & \textbf{5} & \textbf{5} & \textbf{5} & \textbf{5} & \textbf{5} & \textbf{5} & \textbf{5}  \\
&\textbf{5} & \textbf{5} & \textbf{5} & \textbf{5} & \textbf{5} & \textbf{5} & \textbf{5} & \textbf{5} & \textbf{5}  \\

\multirow{2}{12pt}{$0.1$} & \textbf{6} & \textbf{6} & 5     & \textbf{6} & \textbf{6} & \textbf{6}      & 5     & 5     & 5 \rule{0pt}{20pt} \\
&\textbf{6} & \textbf{6} & \textbf{6} & \textbf{6} & \textbf{6} & \textbf{6} & \textbf{6} & 5     & 5 \\

\multirow{2}{12pt}{$0.2$} & \textbf{6} & \textbf{6} & \textbf{6} & \textbf{6} & \textbf{6} & \textbf{6} & \textbf{6} & 5     & 5 \rule{0pt}{20pt} \\
&\textbf{6} & \textbf{6} & \textbf{6} & \textbf{6} & \textbf{6} & \textbf{6} & \textbf{6} & \textbf{6}  & 5 \\

\multirow{2}{12pt}{$0.3$} &  \textbf{7} & \textbf{7} & 6     & \textbf{7} & \textbf{7} & \textbf{7} & 6     & 6     & 6 \rule{0pt}{20pt} \\
 &\textbf{7} & \textbf{7} & \textbf{7} & \textbf{7} & \textbf{7} & \textbf{8} & \textbf{7}  & 6     & 6 \\

\multirow{2}{12pt}{$0.4$} &   7     & \textbf{8} & 7     & \textbf{8}  & \textbf{8} & \textbf{8} & \textbf{8} & 7     & 6 \rule{0pt}{20pt} \\
&8     & 8     & 8     & 8     & \textbf{9} & \textbf{9} & 8     & 7     & 7 \\

\multirow{2}{12pt}{$0.5$} &  8     & 8     & 7     & 8     & \textbf{9} & \textbf{9} & 8     & 7     & 7\rule{0pt}{20pt}      \\
 &9     & 9     & 9     & \textbf{10} & \textbf{10}  & \textbf{10} & \textbf{10}  & 8     & 8 \\

\multirow{2}{12pt}{$0.6$} &  9     & 10    & 9     & 9     & \textbf{11} & \textbf{11}     & 10   & 9      & 8 \rule{0pt}{20pt} \\
  & 11    & 12    & 11    & 11    & 12    & \textbf{13}  & 12      & 10    & 10 \\

\multirow{2}{12pt}{$0.7$} &  10    & 11    & 10    & 10    & \textbf{12} & \textbf{13} & \textbf{12}    & 10    & 10 \rule{0pt}{20pt} \\
 &12    & 14    & 12    & 13    & 14    & \textbf{16} & \textbf{15}    & 13     & 12 \\

\multirow{2}{12pt}{$0.8$} &  11    & 13    & 11    & 11    & 14  &  \textbf{16} & \textbf{15} & 12    & 11\rule{0pt}{20pt}  \\
  & 15    & 17    & 15    & 15    & 18   & \textbf{20} & \textbf{20}     & 16    & 16 \\

\multirow{2}{12pt}{$0.9$} & 14    & 15    & 13    & 14    & 16     & \textbf{19} & \textbf{19} & 15    & 14 \rule{0pt}{20pt} \\
   &18    & 20    & 18    & 17    & 21    & \textbf{25} & \textbf{25}  & 20    & 20     \\
\hline \hline
\label{tab5}
\end{tabular} 

\end{table} 


The results shown in Tables \ref{tab2} and \ref{tab3} indicate that all the tests maintain the nominal significance level of 0.05. Furthermore, the results demonstrate that the proposed test $S^{(2)}_{n,3,2}$ outperforms all the other tests for the majority of alternatives considered, closely followed by $S^{(1)}_{n,3,2}$ and $G_{n,2}$. The powers of $T^{(1)}_{n,3,2}$ and $T^{(2)}_{n,3,2}$ do not compare favorably to those of $S^{(1)}_{n,3,2}$ and $S^{(2)}_{n,3,2}$ but are still competitive in terms of power against the traditional tests; i.e., $KS_n$, $CV_n$ and $AD_n$. We also note that the tests $T^{(1)}_{n,3,2}$ and $T^{(2)}_{n,3,2}$ produce the highest powers against the $LN(2.5)$ alternative by a substantial margin. The results obtained for the mixture distributions are in accordance with those associated with the fixed alternatives.

\newpage

\section{Practical application}
\label{Sect5}

Below, we apply each of the tests considered to an observed data set. The data concerned is the the lifetime tournament earnings, up to 1980, of all professional golfers whose earnings exceeded \$700 000, as reported in the Golf magazine, 1981 yearbook. This data set was also discussed and analysed by \cite{Arn2015}. The reported salaries, in thousands of dollars, can be found in Table \ref{tab6}.

\begin{table}
\begin{centering}
\begin{tabular}{cccccccccc}
\hline 
708 & 712 & 729 & 746 & 753 & 759 & 769 & 771 & 778 & 778 \\
814 & 816 & 820 & 825 & 841 & 844 & 849 & 871 & 878 & 883 \\
912 & 944 & 965 & 1001 & 1005 & 1016 & 1031 & 1051 & 1056 & 1066 \\
1092 & 1095 & 1109 & 1171 & 1184 & 1208 & 1338 & 1374 & 1410 & 1433 \\
1519 & 1537 & 1627 & 1684 & 1690 & 1829 & 1858 & 2202 & 2474 & 3581 \\
\hline 
\end{tabular}
\par\end{centering}
\caption{The golfer data set.}
\label{tab6}
\end{table}

The support of the data in Table \ref{tab6} is, per definition, $(700,\infty)$. As a result, we rescale the data by dividing each number in the table by $700$ and then we test the hypothesis that the resulting data are realised from a Pareto distribution. When fitting a Pareto distribution to the data using the method of moments, we obtain $\widehat{\beta}_n=2.495$. Before proceeding to the results pertaining to the formal testing procedures, we consider visual tests of fit for the Pareto distribution. Figure \ref{Dists.pdf} shows the empirical distribution function of the rescaled data, together with the fitted Pareto distribution function. Figure \ref{QQplot.pdf} shows a quantile-quantile plot comparing the empirical quantiles to those of the fitted distribution. Both figures indicate a close correspondence between the empirical properties of the data and those expected under the null hypothesis of the Pareto distribution.

\begin{figure}[tb]
    \centering
    \includegraphics[height=7cm, width=12cm]{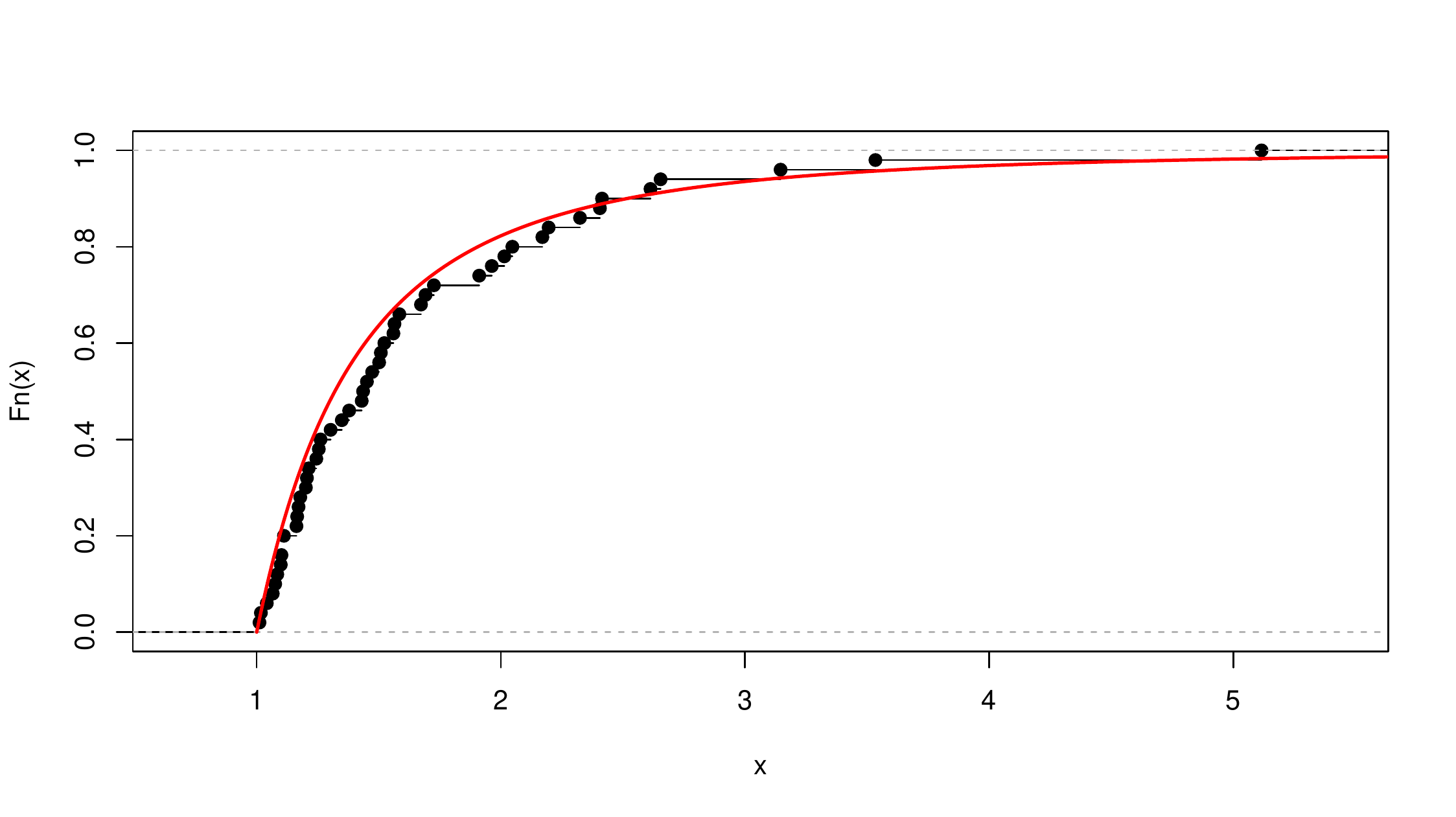}
    \caption{Empirical and fitted distribution functions.}
    \label{Dists.pdf}
\end{figure}

\begin{figure}[tb]
    \centering
    \includegraphics[height=7cm, width=10cm]{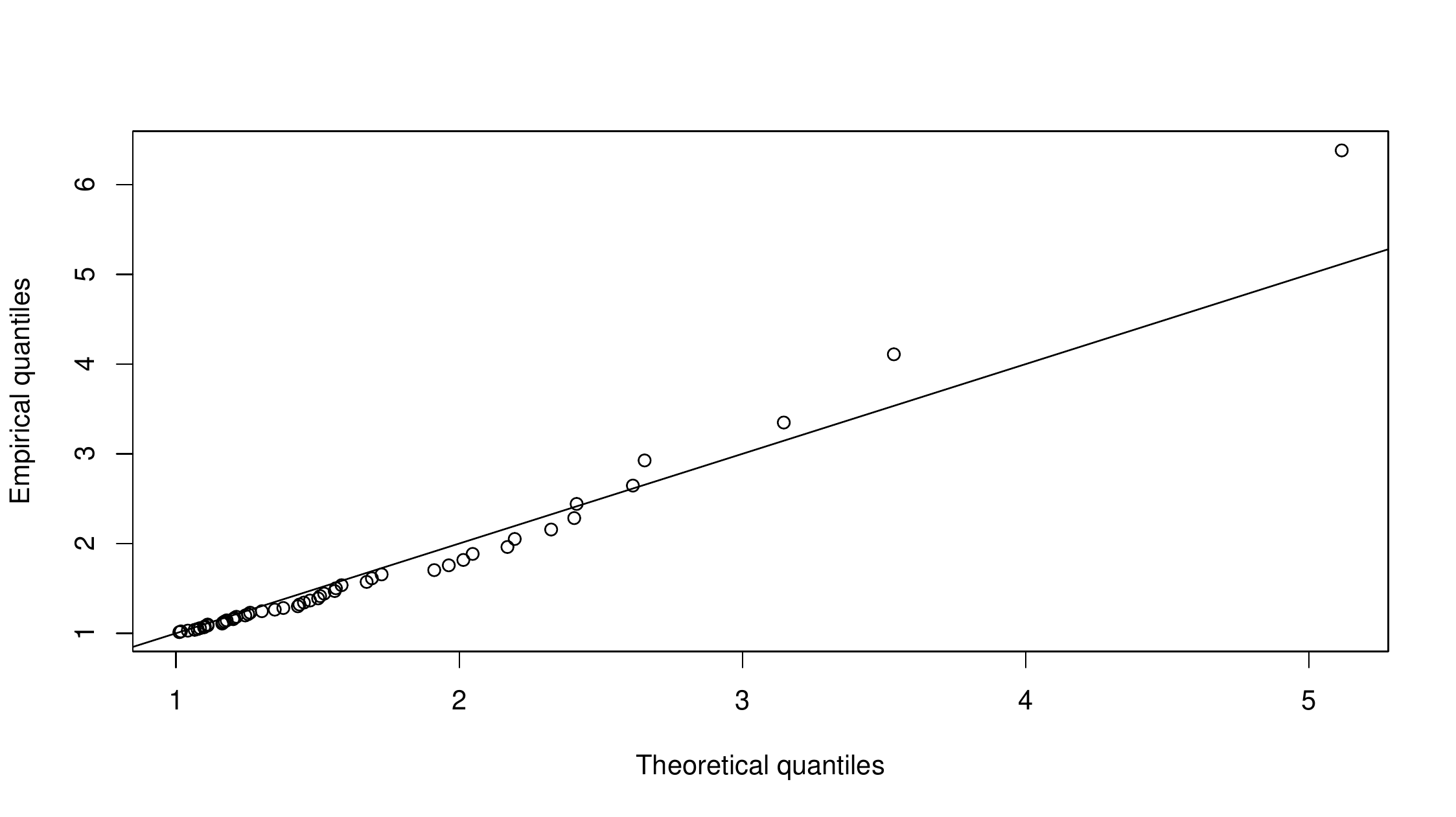}
    \caption{Quantile-quantile plot.}
    \label{QQplot.pdf}
\end{figure}

We now turn our attention to the results obtained using the goodness-of-fit tests discussed in Section \ref{Sect4}. Table \ref{tab7} contains the test statistic values with the corresponding $p$-values of the tests. These $p$-values were calculated based on 10 \ 000 samples of size 50 simulated from a Pareto distribution with parameter $\beta=2.495$ (corresponding to the parameter of the fitted distribution).

\begin{table}
\begin{centering}
\begin{tabular}{cllcccll}
\hline 
Test & Statistic  & $p$-value &&& Test & Statistic  & $p$-value\tabularnewline
\hline 
\hline 
$KS_n$ & 0.125 & 0.3211 &&& $S^{(1)}_{n,3,2}$ & $4 \times 10^{-3}$ & 0.2245 \\
$CV_n$ & 0.158 & 0.2873 &&& $S^{(2)}_{n,3,2}$ & $3\times10^{-3}$ & 0.1929 \\
$AD_n$ & 3.433 & 0.2857 &&& $T^{(1)}_{n,3,2}$ & $2\times10^{-3}$ & 0.3311 \\
$ZA_n$ & 39.332 & 0.0991 &&& $T^{(2)}_{n,3,2}$ & $2\times10^{-3}$ & 0.2869 \\
$G_{n,2}$ & 0.792 & 0.1783 \\
\hline 
\end{tabular}
\par\end{centering}
\caption{Summary results for the lifetime tournament earnings through golf.}
\label{tab7}
\end{table}


It is clear from the reported $p$-values in Table \ref{tab7} that none of the tests considered reject the assumption that the data are realised from a Pareto distribution at a 5\% level of significance. The is in accordance with the findings of \cite{Arn2015}.

\section{Concluding remarks}
\label{Sect6}

In this study, we introduced new classes of goodness-of-fit tests for the Pareto distribution based on a characterisation involving the sample minimum. The proposed tests are based on empirical characteristic functions estimated via $V$ and $U$ statistics, and include two tuning parameters. We compared the performances of our tests to those of some commonly used goodness-of-fit tests and the results presented suggest that the newly proposed tests are competitive, often outperforming the existing tests against the alternative distributions considered. Based on the numerical performance of the tests, we recommend using the tests based on $V$ statistics (using a Gaussian kernel and setting the tuning parameters to $m=3$ and $a=2$; this test is denoted $S^{(2)}_{n,3,2}$ in the text).

We mentioned in the introduction that the Pareto distribution has found applications in the field of survival analysis and reliability theory. Censoring often occurs in these fields due to the nature of the study. A possible avenue for future research is to modify the newly proposed statistics to test for censored Pareto distributions. Some work in this regards has been done by \cite{fernandez2020kaplan}, where they consider Kaplan-Meier $U$ and $V$ statistics, and by \cite{cuparic2022new}, where they derive asymptotic results for a new characterisation based test for exponentiality in the presence of random right censored data. Further details regarding the development of tests for the cencored exponential distribution can be found in \cite{bothma2021exponential}, while \cite{bothma2021weibull} considers tests for the censored Weibull distribution. For results pertaining specifically to tests for the censored Pareto distribution, the interested reader is referred to \cite{Lethani}.

\bibliographystyle{apa-good}
\bibliography{references}

\begin{thebibliography}{28}
\expandafter\ifx\csname natexlab\endcsname\relax\def\natexlab#1{#1}\fi
\expandafter\ifx\csname url\endcsname\relax
  \def\url#1{{\tt #1}}\fi
\expandafter\ifx\csname urlprefix\endcsname\relax\def\urlprefix{URL }\fi

\bibitem[{Allison et~al.(2022)Allison, Milo{\v{s}}evi{\'c}, Obradovi{\'c}, \&
  Smuts}]{allison2021distribution}
Allison, J.~S., Milo{\v{s}}evi{\'c}, B., Obradovi{\'c}, M., \& Smuts, M.
  (2022).
\newblock Distribution-free goodness-of-fit tests for the {P}areto distribution
  based on a characterization.
\newblock {\em Computational Statistics\/}, {\em 37\/}, 403--418.

\bibitem[{Amin(2007)}]{Ami2007}
Amin, Z.~H. (2007).
\newblock Tests for the validity of the assumption that the underlying
  distribution of life is {P}areto.
\newblock {\em Journal of Applied Statistics\/}, {\em 34\/}(2), 195--201.

\bibitem[{Arnold(2015)}]{Arn2015}
Arnold, B.~C. (2015).
\newblock {\em Pareto {D}istributions\/}.
\newblock New York: CRC Press.

\bibitem[{Baringhaus et~al.(2017)Baringhaus, Ebner, \&
  Henze}]{baringhaus2017limit}
Baringhaus, L., Ebner, B., \& Henze, N. (2017).
\newblock The limit distribution of weighted ${L}^{2}$-goodness-of-fit
  statistics under fixed alternatives, with applications.
\newblock {\em Annals of the Institute of Statistical Mathematics\/}, {\em
  69\/}(5), 969--995.

\bibitem[{Baringhaus \& Henze(1988)}]{baringhaus1988consistent}
Baringhaus, L., \& Henze, N. (1988).
\newblock A consistent test for multivariate normality based on the empirical
  characteristic function.
\newblock {\em Metrika\/}, {\em 35\/}(1), 339--348.

\bibitem[{Beirlant et~al.(2004)Beirlant, Goegebeur, Segers, \&
  Teugels}]{BeiGoeSegTeu}
Beirlant, J., Goegebeur, Y., Segers, J., \& Teugels, J. (2004).
\newblock {\em Statistics of {E}xtremes: {T}heory and {A}pplications\/}.
\newblock Chichester: John Wiley and Sons.

\bibitem[{Bothma et~al.(2021)Bothma, Allison, Cockeran, \&
  Visagie}]{bothma2021exponential}
Bothma, E., Allison, J.~S., Cockeran, M., \& Visagie, I. J.~H. (2021).
\newblock Characteristic function and {L}aplace transform-based tests for
  exponentiality in the presence of random right censoring.
\newblock {\em Stat\/}, {\em 10\/}(1), e394.

\bibitem[{Bothma et~al.(2022)Bothma, Allison, \& Visagie}]{bothma2021weibull}
Bothma, E., Allison, J.~S., \& Visagie, I. J.~H. (2022).
\newblock New classes of tests for the {W}eibull distribution using {S}tein's
  method in the presence of random right censoring.
\newblock {\em Computational Statistics\/}.

\bibitem[{Bourguignon et~al.(2016)Bourguignon, Saulo, \& Fernandez}]{BSF2016}
Bourguignon, A., Saulo, B., \& Fernandez, R.~N. (2016).
\newblock A new {P}areto-type distribution with applications in reliability and
  income data.
\newblock {\em Physica A: Statistical Mechanics and its Applications\/}, {\em
  457\/}, 166--175.

\bibitem[{Chu et~al.(2019)Chu, Dickin, \& Nadarajah}]{CDN2019}
Chu, J., Dickin, S., \& Nadarajah, S. (2019).
\newblock {A review of goodness of fit tests for Pareto distributions}.
\newblock {\em Journal of Computational and Applied Mathematics\/}, {\em
  361\/}, 13--41.

\bibitem[{Cupari{\'c} \& Milo{\v{s}}evi{\'c}(2022)}]{cuparic2022new}
Cupari{\'c}, M., \& Milo{\v{s}}evi{\'c}, B. (2022).
\newblock New characterization-based exponentiality tests for randomly censored
  data.
\newblock {\em Test\/}, {\em 31\/}(2), 461--487.

\bibitem[{Fern{\'a}ndez \& Rivera(2020)}]{fernandez2020kaplan}
Fern{\'a}ndez, T., \& Rivera, N. (2020).
\newblock Kaplan-{M}eier {V}-and {U}-statistics.
\newblock {\em Electronic Journal of Statistics\/}, {\em 14\/}(1), 1872--1916.

\bibitem[{Feuerverger \& Mureika(1977)}]{feuerverger1977empirical}
Feuerverger, A., \& Mureika, R.~A. (1977).
\newblock The empirical characteristic function and its applications.
\newblock {\em The Annals of Statistics\/}, {\em 5\/}(1), 88--97.

\bibitem[{Fisk(1961)}]{Fis1961}
Fisk, P.~R. (1961).
\newblock The graduation of income distributions.
\newblock {\em Econometrica\/}, {\em 29\/}(2), 171--185.

\bibitem[{Giacomini et~al.(2013)Giacomini, Politis, \&
  White}]{giacomini2013warp}
Giacomini, R., Politis, D.~N., \& White, H. (2013).
\newblock A warp-speed method for conducting {M}onte {C}arlo experiments
  involving bootstrap estimators.
\newblock {\em Econometric Theory\/}, {\em 29\/}(3), 567--589.

\bibitem[{Gupta(1973)}]{gupta1973characteristic}
Gupta, R.~C. (1973).
\newblock A characteristic property of the exponential distribution.
\newblock {\em Sankhy{\=a}: The Indian Journal of Statistics, Series B\/}, (pp.
  365--366).

\bibitem[{Isma{\"\i}l(2004)}]{Ism2004}
Isma{\"\i}l, S. (2004).
\newblock A simple estimator for the shape parameter of the {P}areto
  distribution with economics and medical applications.
\newblock {\em Journal of Applied Statistics\/}, {\em 31\/}(1), 3--13.

\bibitem[{Klar \& Meintanis(2005)}]{klar2005tests}
Klar, B., \& Meintanis, S.~G. (2005).
\newblock Tests for normal mixtures based on the empirical characteristic
  function.
\newblock {\em Computational Statistics \& Data Analysis\/}, {\em 49\/}(1),
  227--242.

\bibitem[{Meintanis(2009)}]{meintanis2009unified}
Meintanis, S.~G. (2009).
\newblock A unified approach of testing for discrete and continuous {P}areto
  laws.
\newblock {\em Statistical Papers\/}, {\em 50\/}(3), 569--580.

\bibitem[{Meintanis(2016)}]{Mei2016}
Meintanis, S.~G. (2016).
\newblock A review of testing procedures based on the empirical characteristic
  function.
\newblock {\em The South African Statistical Journal\/}, {\em 50\/}(1), 1--14.

\bibitem[{Ndwandwe et~al.(2022)Ndwandwe, Allison, Santana, \&
  Visagie}]{ndwandwe2022testing}
Ndwandwe, L., Allison, J.~S., Santana, L., \& Visagie, I. J.~H. (2022).
\newblock Testing for the {P}areto type {I} distribution: {A} comparative
  study.
\newblock {\em arXiv preprint arXiv:2211.10088\/}.

\bibitem[{Ndwandwe et~al.(2021)Ndwandwe, Allison, \& Visagie}]{Lethani}
Ndwandwe, L., Allison, J.~S., \& Visagie, I. J.~H. (2021).
\newblock A new fixed point characterisation based test for the {P}areto
  distribution in the presence of random censoring.
\newblock {\em Proceedings of the 62nd Annual Conference of SASA\/}, (1),
  17--23.

\bibitem[{Nofal \& El~Gebaly(2017)}]{nofal2017new}
Nofal, Z.~M., \& El~Gebaly, Y.~M. (2017).
\newblock New characterizations of the {P}areto distribution.
\newblock {\em Pakistan Journal of Statistics and Operation Research\/}, {\em
  13\/}(1), 63--74.

\bibitem[{Pareto(1897)}]{Par}
Pareto, V. (1897).
\newblock {\em Cours d’economie Politique, Vol. II.\/}.
\newblock Lausanne: F. Rouge.

\bibitem[{{R Core Team}(2020)}]{rteam}
{R Core Team} (2020).
\newblock {\em R: A Language and Environment for Statistical Computing\/}.
\newblock R Foundation for Statistical Computing, Vienna, Austria.
\newline\urlprefix\url{http://www.R-project.org/}

\bibitem[{Rytgaard(1990)}]{rytgaard1990estimation}
Rytgaard, M. (1990).
\newblock Estimation in the {P}areto distribution.
\newblock {\em ASTIN Bulletin: The Journal of the IAA\/}, {\em 20\/}(2),
  201--216.

\bibitem[{Soliman(2000)}]{Sol2000}
Soliman, A.~A. (2000).
\newblock Bayes prediction in a {P}areto lifetime model with random sample
  size.
\newblock {\em Journal of the Royal Statistical Society: Series D (The
  Statistician)\/}, {\em 49\/}(1), 51--62.

\bibitem[{Zhang(2002)}]{zhang2002powerful}
Zhang, J. (2002).
\newblock Powerful goodness-of-fit tests based on the likelihood ratio.
\newblock {\em Journal of the Royal Statistical Society: Series B (Statistical
  Methodology)\/}, {\em 64\/}(2), 281--294.

\end{thebibliography}

\end{document}